\documentclass[final]{article}
\usepackage[left=1.2in, right=1.5in]{geometry}
\usepackage{./MacrosDaniel}
\usepackage{amssymb,epsfig,amsmath,amsthm}
\usepackage{authblk}
\usepackage{graphicx}
\graphicspath{{./Figures/}}
\usepackage{url}
\usepackage{booktabs}
\usepackage{color}
\usepackage{bm}
\usepackage{enumitem}
\usepackage{breakcites}
\usepackage{multirow}

\newcounter{enumi_saved}
\newenvironment{myenumerateA} {
    \begin{enumerate}[leftmargin=*,label={(A\arabic{enumi})}]\setcounter{enumi}{\value{enumi_saved}}}
    {\setcounter{enumi_saved}{\value{enumi}}\end{enumerate}}

\newcounter{enumii_saved}
\newenvironment{myenumerateB} {
    \begin{enumerate}[leftmargin=*,label={(B\arabic{enumi})}]\setcounter{enumi}{\value{enumii_saved}}}
    {\setcounter{enumii_saved}{\value{enumi}}\end{enumerate}}

\newcounter{enumiii_saved}
\newenvironment{myenumerateC} {
    \begin{enumerate}[leftmargin=*,label={(C\arabic{enumi})}]\setcounter{enumi}{\value{enumiii_saved}}}
    {\setcounter{enumiii_saved}{\value{enumi}}\end{enumerate}}

\begin{document}

\setlength{\parindent}{0pt}
\setlength{\parskip}{1ex plus 0.5ex minus 0.2ex}

\theoremstyle{plain}
\newtheorem{theorem}{Theorem}
\newtheorem{proposition}{Proposition}
\newtheorem{corollary}{Corollary}
\newtheorem{lemma}{Lemma}

\theoremstyle{definition}
\newtheorem{definition}{Definition}
\newtheorem{example}{Example}
\newtheorem{assumption}{Assumption}

\theoremstyle{remark}
\newtheorem{remark}{Remark}

\title{Market Completion with Derivative Securities}

\author{Daniel C. Schwarz}

\affil{Carnegie Mellon University,\\ Department of Mathematical Sciences,\\ 5000 Forbes Avenue,\\ Pittsburgh, PA, 15213, USA}
\date{May 27th, 2015}

\maketitle

\begin{abstract}
Let $S^F$ be a $\mathbb{P}$-martingale representing the price of a primitive asset in an incomplete market framework. We present easily verifiable conditions on model coefficients which guarantee the completeness of the market in which in addition to the primitive asset one may also trade a derivative contract $S^B$. Both $S^F$  and $S^B$ are defined in terms of the solution $X$ to a $2$-dimensional stochastic differential equation: $S^F_t = f(X_t)$ and $S^B_t\df\E[g(X_1) | \mathcal{F}_t]$. From a purely mathematical point of view we prove that every local martingale under $\mathbb{P}$ can be represented as a stochastic integral with respect to the $\mathbb{P}$-martingale $S \df (S^F\ S^B)$. Notably, in contrast to recent results on the endogenous completeness of equilibria markets, our conditions allow the Jacobian matrix of $(f,g)$ to be singular everywhere on $\R^2$. Hence they cover, as a special case, the prominent example of a stochastic volatility model being completed with a European call (or put) option.
\end{abstract}

\vspace{1cm}

\newpage
\section{Introduction}\label{sec:intro}
Let $(\Omega,\mathcal{F},\mathbb{P})$ be a probability space, consider a fixed time horizon equal to one and let $\mathbf{F} = (\mathcal{F}_t)_{t\in[0,1]}$ be a filtration satisfying the usual conditions with $\mathcal{F}_0$ containing only $\Omega$ and the null sets of $\mathbb{P}$ and with $\mathcal{F}_1 = \mathcal{F}$. Let $S=(S^j_t)$  be a $d$-dimensional stochastic process describing the evolution of the \textit{discounted} prices of liquidly traded securities in a financial market and with the property that $S$ is a (vector) martingale under the measure $\mathbb{P}$. The model is said to be \textit{complete}, if any contingent claim payoff can be obtained as the terminal value of a self-financing trading strategy. The 2nd fundamental theorem of asset pricing (cf. \cite{jHarrison1983}) allows us to restate the completeness property in purely mathematical terms as follows: every local martingale $M=(M_t)$ admits an integral representation with respect to $S$, that is,
\begin{equation}\label{eq:int_rep}
 M_t = M_0 + \int_0^t H_u\ \Id S_u,\quad t\in [0,1],
\end{equation}
for some predictable $S$-integrable process $H=(H^j_t)$. The 2nd fundamental theorem of asset pricing also asserts that the above statements are equivalent to $\mathbb{P}$ being the unique
martingale measure for $S$ in the class of equivalent measures.

The process $S$ may for example describe the prices of stocks or option contracts, which nowadays are often traded as liquidly as their underlyings. Depending on the application one has in mind the construction of $S$ differs significantly. In general there are three possibilities to consider. Given its initial value, $S$ may be defined in a \textit{forward} form, in terms
of its predictable characteristics under the measure $\mathbb{P}$. In this case the
verification of the completeness property is straightforward. For
example, if $S$ is a driftless diffusion process under the measure $\mathbb{P}$ with volatility matrix-process $\sigma =(\sigma_t)$, then the market is complete if and only if $\sigma$ has full rank $\Id \mathbb{P}\times \Id t$ almost surely (cf. \cite[Theorem 6.6]{iKaratzas1998a}). Alternatively, $S$ can be defined in a \textit{backward} form, as the
conditional expectation under $\mathbb{P}$ of its given terminal
value. Finally, some components of $S$ may be defined in a forward form and others in a backward form leading to a \textit{forward-backward} setup.

In the present paper we assume the last setup above and focus on the case of two dimensions that is $d=2$. In particular, let $S^F = (S^F_t)$ and $S^B = (S^B_t)$ be scalar-valued martingales under $\mathbb{P}$, such that
\begin{equation*}
S = \left(\begin{array}{c} S^F \\ S^B \end{array}\right).
\end{equation*}
One may view the forward component as the discounted price of a primitive asset and the
backward component as that of a derivative security. That is, given a process $\sigma^F = (\sigma^{F,j}_t)_{j=1,2}$, a $\mathbb{P}$-Brownian
motion $W = (W^j_t)_{j=1,2}$ and a random variable $\psi$, the processes $S^F$ and $S^B$ are defined by
\begin{align*}
S^{F}_t &= S^{F}_0 + \int_0^t \sigma^F_u\ \Id
 W_u,\\
S^B_t &\df \E[\psi|\mathcal{F}_t], \quad \text{for } t \in [0,1].
\end{align*}
We are looking for easily verifiable conditions on $\sigma^F$ and $\psi$,
guaranteeing the integral representation property of all
$\mathbb{P}$-martingales with respect to $S$ and hence the
completeness of the market, in which, in addition to the primitive asset $S^F$, also the derivative contract $S^B$ can be traded.

In principal the proof of our main result Theorem \ref{thm:FBMR} below generalizes to the $d$-dimensional case. The reason we present the two-dimensional case only is twofold: first, the structural conditions on the coefficients $\sigma^F$ and $\psi$ become very complex in higher dimensions; second, using our current methods an extension to higher dimensions would require additional regularity of $\psi$ and, in particular, exclude the payoff functions of call and put options, which are only once weakly differentiable.

For our analysis we assume that $\sigma^F$ and $\psi$ are specified in terms of a solution $X$ to a two-dimensional stochastic differential equation with drift
vector $b=b(t,x)$ and volatility matrix $\sigma = \sigma(t,x)$. With respect to the space variable our conditions are quite classical: $b=b(t,\cdot)$ is once continuously differentiable and $\sigma = \sigma(t,\cdot)$ is twice continuously differentiable and possesses a bounded inverse. Further, the functions themselves and their derivatives are bounded. With respect to time our conditions are quite exacting: $b=b(\cdot,x)$ and $\sigma = \sigma(\cdot,x)$ have to be real analytic on $(0,1)$.

Our results extend and rigorously prove ideas on the completion of
markets with derivative securities, first formulated in \cite{mRomano1997} and \cite{mDavis2008}.

The paper \cite{mRomano1997} is concerned with the specific case of stochastic volatility
models. The main result in this paper requires the derivative payoff function to be a convex function of the stock price only and, unless given by the special case of a European call or put option, to be twice continuously differentiable. Perhaps most limiting from the point of view of applicability, it is required that the volatility risk premium is such that the drift coefficient of the volatility process under the equivalent martingale measure does not depend on the stock price. Moreover, also the correlation between the asset price and its (stochastic) volatility process and the volatility of the volatility process must not depend on the stock price.

In \cite{mDavis2008} the setup is not restricted to the two-dimensional case. However, the key conditions in this paper are not placed on model primitives, but on the conditional expectation $\E[(S^F_1,\psi)^\star|\mathcal{F}_t] = v(t,X_t)$. In particular, $v$ is
\textit{assumed} to be (jointly) real analytic in the time and space variables and, in the main theorem of the paper, the Jacobian matrix (with respect to $x$) of $v=v(t,x)$ is assumed to be nonzero on some open subset of $(0,1)\times\R^2$.

Our work is intimately related to recent results on the integral
representation of martingales, which were motivated by the problem of
the endogenous completeness of continuous-time Radner equilibria in financial economics
(cf. \cite{dKramkov2012,jHugonnier2012,fRiedel2012,rAnderson2008}). The differences between these results and ours are twofold: first, in a Radner equilibrium setting $S$ is specified purely in a backward form and the setup does not accommodate a forward component; second and perhaps most important, if $S^F_1 = f(X_1)$ and $\psi = g(X_1)$ the aformentioned results require the Jacobian matrix
\begin{equation*}
\left(
 \begin{array}{cc}
 f_{x^1} & f_{x^2}\\
 g_{x^1} & g_{x^2}
 \end{array}\right)(x), \quad x\in \R^2,
\end{equation*}
to have full rank at least on some open subset of $\R^2$. This condition is not satisfied even in the most pivotal example of the completion of a stochastic volatility model with a European call or put option, where the corresponding Jacobian matrix is singular everywhere on $\R^2$. We replace this requirement with a novel condition involving aside from $f$ and $g$ also the coefficients of the state process $b$ and $\sigma$ and which is satisfied in the aforementioned example of a typical stochastic volatility model being completed with a European call option (see Section \ref{sec:example}).

At first sight it may appear that the most restrictive condition, limiting the
applicability of our result, is the boundedness assumption on the
coefficients of the diffusion $X$. This assumption stems from the theory of elliptic and
parabolic partial differential equations, which plays an essential part in our proofs. However, we demonstrate in Section \ref{sec:example} how we can still accommodate popular models from financial mathematics such as geometric Brownian motion or mean-reverting processes by means of suitable changes of variables.

\paragraph{Notation and basic concepts}
Let $\mathbf{X}$ be a Banach space with norm
$\|\cdot\|$. In the sequel we will frequently use maps $h:[0,1]\to \mathbf{X}$, which are \textit{H\"{o}lder continuous} on $[0,1]$, that is, there exist constants $N>0$ and $\delta>0$ such that
\begin{equation*}
 \|h(u)-h(t)\| \leq N|u-t|^\delta, \quad u,t\in [0,1],
\end{equation*}
and \textit{analytic} on $(0,1)$, that is, for every
$u\in (0,1)$ there exist $\epsilon(u)>0$ and a family $\{A_n(u)\}$ of elements in $\mathbf{X}$, such that
\begin{equation*}
 h(t) = \sum_{n=0}^\infty A_n(u)(t-u)^n, \quad t\in (0,1),\ |t-u|<\epsilon (u).
\end{equation*}

For multi-indices $\alpha = (\alpha_1,\ldots,\alpha_d)$ of nonnegative integers, we use the notation convention $|\alpha|\df\sum_{i=1}^d \alpha_i$ and
\begin{equation*}
D^\alpha \df \frac{\partial^{|\alpha|}}{\partial
  x^{\alpha_1}_1\ldots \partial x^{\alpha_d}_d}.
\end{equation*}
Let $U\subset \R^d$. Throughout the text the following spaces will be used:

$\mathbf{L}_p(U)$ (for $p\geq 1$): the Lebesgue space of Lebesgue-measurable, real-valued functions $h$ on $U$ with the norm $\|h\|_{\mathbf{L}_p(U)}\df (\int_U |h|^p\ \Id x )^{1/p}$; $\mathbf{L}_p \df \mathbf{L}_p(\R^2)$.

$\mathbf{L}_\infty(U)$: the Lebesgue space of essentially bounded,
real-valued functions $h$ on $U$ with the norm
$\|h\|_{\mathbf{L}_\infty(U)}\df \text{ess} \sup_{U}|h|$; $\mathbf{L}_\infty \df \mathbf{L}_\infty(\R^2)$.

$\mathbf{C}^k(U)$: the Banach space of all $k$-times continuously differentiable, real-valued functions $h$ on $U$ with the norm
\begin{equation*}
\| h\|_{\mathbf{C}^k(U)} = \|h\|_{\mathbf{C}(U)} + \sum_{1\leq |\alpha|\leq k}\|D^\alpha u\|_{\mathbf{C}(U)}, 
\end{equation*}
where $\|h\|_{\mathbf{C}(U)}\df \sup_{U}|h|$; $\mathbf{C}^k \df \mathbf{C}^k(\R^2)$.

%
%
Recall that a locally integrable function $h$ on $U$ is weakly
differentiable, if for every index $j=1,\ldots,d$ there exists a locally integrable function $g^j$ such that the identity
\begin{equation*}
\int_{U} g^j(x)\varphi(x)\ \Id x = -\int_{U} h(x) \frac{\partial
  \varphi}{\partial x^j}(x)\ \Id x
\end{equation*}
holds for every function $\varphi$ belonging to $\mathbf{C}^\infty_0(U)$, the space of infinitely many times differentiable functions with compact support in $U$. In this case we define $h_{x^j}\df g^j$. Weak derivatives of higher orders are defined recursively.

As is common, for $p\geq 1$, we denote by $p'$ the conjugate exponent
of $p$, defined by $p' \df p/(p-1)$ for $1<p<\infty$, $p'\df \infty$, if
$p=1$ and $p'\df 1$, if $p=\infty$.

With these definitions in mind we define the following spaces:

$\mathbf{W}^m_p(U)$ (for $m\in \{0,1,\ldots\}$ and $p\geq 1$): the
Banach space of $m$-times weakly differentiable functions $h$ with the norm
\begin{equation*}
\| h \|_{\mathbf{W}^m_p(U)} \df \|h\|_{\mathbf{L}_p(U)} + \sum_{1\leq
  |\alpha|\leq m} \| D^\alpha h\|_{\mathbf{L}_p(U)};
\end{equation*}
(The case $m=0$ recovers the classical Lebesgue spaces $\mathbf{L}_p(U)$.) $\mathbf{W}^m_p\df \mathbf{W}^m_p(\R^2)$.

$\mathbf{W}^m_{p,0}(U)$ (for $m\in \{0,1,\ldots\}$ and $p\geq 1$): the
Banach space obtained by taking the closure of $\mathbf{C}^\infty_0(U)$
in the space $\mathbf{W}^m_p(U)$.

$\mathbf{W}^{-m}_p(U)$ (for $m\in \{0,1,\ldots\}$ and $p \geq 1$): the
Banach space of all distributions $h$ of the form 
\begin{equation}\label{eq:dual_decomp}
h = \sum_{0\leq
  |\alpha|\leq m} (-1)^{|\alpha|}\langle D^\alpha \cdot, u_\alpha\rangle,
\end{equation}
where $\langle \cdot,\cdot \rangle$
denotes the inner product in $\mathbf{L}_2$ and $u_\alpha\in \mathbf{L}_{p}(U)$, with the norm
\begin{equation*}
\|h\|_{\mathbf{W}^{-m}_p(U)} \df \min \{\sum_{0\leq |\alpha|\leq m} \|
u_{\alpha}\|_{\mathbf{L}_{p}(U)} : u \text{ satisfies }
\eqref{eq:dual_decomp}\}.
\end{equation*}

For $T\subset \R$, we also define $\mathbf{W}^{r,m+2r}_p(T\times U)$ (for $r,m\in \{0,1,\ldots\}$ and $p\geq 1$): the
Banach space of functions $h=h(t,x)$, $r$-times weakly differentiable in $t$ and $(m+2r)$-times weakly differentiable in $x$ with the norm
\begin{equation*}
\| h \|_{\mathbf{W}^{r,m+2r}_p(T\times U)} \df \sum_{\substack{|\alpha| + 2\rho \leq m+2r\\ \rho \leq r}} \| D^\alpha \partial_t^\rho h\|_{\mathbf{L}_p(T\times U)}.
\end{equation*}

Our notation is in agreement with standard notation from linear
algebra. Given two vectors $x,y$ in $\R^d$, $xy$ denotes the scalar
product and $|x|\df\sqrt{xx}$. Given a matrix $M\in\R^{m\times n}$
with $m$ rows and $n$ columns, $Mx$ denotes its product with the
column vector $x$, $M^\star$ its transpose and $\|M\|_F \df \sqrt{tr(MM^\star)}$. For an $n\times n$ matrix $M$ we denote the determinant of $M$ either by $|M|$ or by $\det{M}$.
Let $l = (l_1,\ldots,l_k)$ denote a multiindex complying with the condition $1\leq
l_1<\ldots<l_k\leq d$. Given $n\times n$ matrices $M$, $C^1,\ldots,C^k$, we
write $M(l; C^1,\ldots,C^k)$ for the matrix that
is obtained from $M$ by replacing the $l_p$th column of $M$ by the
$l_p$th column of $C^p$, for $p=1,\ldots,k$, while keeping the
remaining columns unchanged; if $k>n$, $M(l; C^1,\ldots,C^k)\df
0$. Let $A$ be an operator on a Banach space $\mathbf{X}$ and $M$ an $n\times n$
matrix such that $m^{ij}$ is in the domain of $A$, $i,j=1,\ldots,n$. We write
$AM$ for the entrywise application of the operator $A$; to wit $AM \df (Am^{ij})_{i,j=1,\ldots,n}$.

For a suitably regular function $h=h(t,x):T\times \R^d\to\R^n$ we denote by $J[h]=J[h](t,x)$ the Jacobian matrix-function of the vector-valued function $h(t,\cdot)$:
\begin{equation*}
 J[h](t,x) \df 
 \left(
 \begin{array}{c}
  \nabla_x h^1\\
  \vdots\\
  \nabla_x h^n
 \end{array}
 \right)(t,x),
 \quad (t,x)\in T\times\R^d,
\end{equation*}
where $\nabla_xh$ is the gradient vector of $h(t,\cdot)$, that is $\nabla_xh \df (\partial_{x^1}h,\ldots,\partial_{x^d}h)$. Similarly, for a suitably regular function $h=h(t,x):T\times \R^d\to\R$ we denote by $H[h]=H[h](t,x)$ the Hessian matrix-function of the scalar-valued function $h(t,\cdot)$, that is $H[h](t,x) \df J[\nabla_x h^\star](t,x)$, for $(t,x)\in T\times\R^d$

Throughout the text $N>0$ denotes a constant, the value of which may vary from line to line.

\section{Main Result: Forward-Backward Martingale Representation}\label{sec:MR}

Let $\R^d$ denote a $d$-dimensional Euclidean vector space and $b = b(t,x) :
[0,1]\times\R^2 \to \R^2$ and $\sigma = \sigma(t,x) : [0,1]\times\R^2
\to \R^{2\times 2}$ measurable functions, which for all
$i,j=1,2$, satisfy the following assumption:
\begin{myenumerateA}
\item \label{as:drift_vol_coeffs}
The maps $t\mapsto b^j(t,\cdot)$ and $t\mapsto \sigma^{ij}(t,\cdot)$ of $[0,1]$ to $\mathbf{C}$ are H\"{o}lder continuous and their restriction to $(0,1)$ is analytic. The map $t\mapsto \sigma^{ij}(t,\cdot)$ is continuous of $[0,1]$ to $\mathbf{C}^2$ and the map $t\mapsto b^j(t,\cdot)$ is continuous of $[0,1]$ to $\mathbf{C}^1$.
 The matrix $\sigma$ is invertible and there exists a constant $N>0$ such that
 \begin{equation}\label{eq:uniform_ell}
  \|\sigma^{-1}(t,x)\|_F\leq N, \quad (t,x) \in [0,1]\times\R^2.
 \end{equation}
\end{myenumerateA}
\begin{remark}
 Note that \eqref{eq:uniform_ell} is equivalent to the uniform ellipticity of the covariance matrix-function $a\df \sigma \sigma^\star$:
 \begin{equation*}
  ya(t,x)y = \|\sigma^\star(t,x) y\|_F^2 \geq \frac{1}{N^2}|y|^2,\quad y\in \mathbb{R}^2,\ (t,x)\in[0,1]\times\mathbb{R}^2.
 \end{equation*}
\end{remark}

Let $X_0\in\R^2$. The assumptions on $b$ and $\sigma$ in \ref{as:drift_vol_coeffs} imply that given a complete, filtered probability space $(\Omega,\mathcal{F}_1,\mathbf{F}=(\mathcal{F}_t)_{t\in[0,1]},\mathbb{P})$ on which is defined a Brownian motion $W$ with values in $\R^2$, there exists a unique stochastic process $X$, also taking values in $\R^2$, such that
\begin{equation}\label{eq:factor_process}
 X_t = X_0 + \int_0^t b(u,X_u)\ \Id u + \int_0^t \sigma(u,X_u)\ \Id W_u, \quad t \in [0,1],
\end{equation}
(cf. \cite[Theorem~2.2, Ch.~5, p.~104]{aFriedman1975}). Here the filtration $\mathbf{F}$ is assumed to be the augmentation of the Brownian filtration, that is,
\begin{equation*}
 \mathcal{F}_t \df \sigma(\mathcal{F}^W_t \cup \mathcal{N}),\quad t\in[0,1],
\end{equation*}
where $\mathcal{F}^W_t$ denotes the $\sigma$-field generated by $(W_u)_{u\in[0,t]}$ and $\mathcal{N}$ denotes the collection of all $\mathbb{P}$-null sets.


Let the measurable function $r:[0,1]\times\R^2\to\R$ satisfy the following:
\begin{myenumerateA}
\item\label{as:tc_dc_coeffs} The map $t\mapsto r(t,\cdot)$ is H\"{o}lder continuous as a map of $[0,1]$ to $\mathbf{C}$, continuous as a map of $[0,1]$ to $\mathbf{C}^1$, analytic as a map of $(0,1)$ to $\mathbf{C}$. The function $r$ is nonnegative:
\begin{equation*}
 r(t,x)\geq 0, \quad (t,x)\in [0,1]\times\R^2.
\end{equation*}
\end{myenumerateA}

Let the measurable function $f=f(t,x):[0,1]\times\R^2\to\R$ be three times weakly differentiable with respect to $x$ and assume that there exists a constant $N>0$ such that, for $j,k,l=1,2$, it holds that:
\begin{myenumerateA}
\item\label{as:A3} The map $t\mapsto e^{-N|\cdot|}\partial_{x^jx^k}f(t,\cdot)$ of $(0,1)$ to $\mathbf{L}_\infty$ is analytic, the map $t\mapsto e^{-N|\cdot|}\partial_{x^j}f(t,\cdot)$ of $[0,1]$ to $\mathbf{L}_\infty$ is continuously differentiable and the map $t\mapsto e^{-N|\cdot|}\partial_{x^jx^kx^l}f(t,\cdot)$ of $[0,1]$ to $\mathbf{L}_\infty$ is continuous.
\end{myenumerateA}
%

Recall that $a \df \sigma\sigma^\star$ is the covariance function of $X$. We denote by $\mathcal{L}^X(t)$, $t\in[0,1]$, the infinitesimal generator of the process $X$:
\begin{align*}
 \mathcal{L}^X(t) &\df \frac{1}{2}\sum_{j,k=1}^2 a^{jk}(t,x)\frac{\partial^2}{\partial x^j \partial x^k} + \sum_{j=1}^2 b^j(t,x)\frac{\partial}{\partial x^j}
\end{align*}
and define the functions $A=A(t,x)$, $B=B(t,x)$ and $C = C(t,x)$ on $[0,1]\times\R^2$ by
\begin{equation*}
\begin{aligned}
 A^{jk} &\df |J[f,a^{jk}]|-2(-1)^j(H[f]a)^{(3-j)k},\\
 B^j &\df |J[f,b^j]| -(-1)^j(\partial_t + \mathcal{L}^X(t) - r)\partial_{x^{(3-j)}}f,\\
 C &\df |J[f,r]|,
 \end{aligned}
\end{equation*}
for $j,k=1,2$.

For suitably regular functions $v=v(x)$, $\varphi = \varphi(x)$ on $\R^2$, for a bounded, open set $K$ in $\R^2$ and for $t\in[0,1]$, we define the pairing
\begin{equation*}
\mathcal{B}_K[v,\varphi;t] \df \int_K \frac{1}{2}\sum_{j,k=1}^2A^{jk}(t,x)\frac{\partial v}{\partial x^j}\frac{\partial \varphi}{\partial x^k}\\
- \sum_{j=1}^2\left(B^j - \frac{1}{2} \sum_{k=1}^2\frac{\partial A^{jk}}{\partial x^k}\right)(t,x)\frac{\partial v}{\partial x^j}\varphi+C(t,x)v\varphi\ \Id x.
\end{equation*}

Let the measurable function $g = g(x) :\R^2\to\R$ be once weakly differentiable and assume that there exists a constant $N>0$ such that:
\begin{myenumerateA}
\item\label{aas:tc_sing_onedim}
Either the Jacobian matrix $J[f,g](1,\cdot)$ has full rank almost everywhere on $\mathbb{R}^2$ or, for every bounded, open set $K$ in $\R^2$ there exists a function $\varphi=\varphi(x)$ belonging to $\mathbf{W}^1_{p,0}(K)$, for some $p\geq1$, such that $\mathcal{B}_{K}[g,\varphi;1]\neq 0$ and
\begin{equation*}
\left| \frac{\partial g}{\partial x^j}(x)\right| \leq e^{N(1+|x|)}, \quad x\in \R^2,\ j=1,2.
\end{equation*}
\end{myenumerateA}

Given the above definitions we define
the scalar-valued random variable $\psi$ by
\begin{equation*}
 \psi \df g(X_1)e^{-\int_0^1r(t,X_t)\ \Id t}.
\end{equation*}

The main result of the paper is
\begin{theorem}[Forward-Backward Martingale Representation]\label{thm:FBMR}
Suppose that \ref{as:drift_vol_coeffs}, \ref{as:tc_dc_coeffs}, \ref{as:A3} and \ref{aas:tc_sing_onedim} hold. Then
the solution $(S^F,S^B,Z)$ to the forward-backward stochastic differential equation
\begin{equation}\label{eq:FB_asset_evol}
\left\{
 \begin{aligned}
  S^F_t &= S^F_0 + \int_0^t e^{-\int_0^u r(s,X_s)\ \Id s}(\nabla_x f\sigma)(u,X_u)\ \Id W_u\\
  S^B_t &= e^{-\int_0^1 r(u,X_u)\ \Id u}g(X_1) - \int_t^1 e^{-\int_0^u r(s,X_s)\ \Id s}Z_u\ \Id W_u
 \end{aligned}
 \right.
\end{equation}
is well-defined. Moreover, every local martingale $M$ under
$\mathbb{P}$ is a stochastic integral with respect to the two-dimensional $\mathbb{P}$-martingale
$S = (S^F_t,S^B_t)$, that is \eqref{eq:int_rep} holds and the market model is complete under $\mathbb{P}$.
\end{theorem}

\begin{remark}
If the function $g=g(x)$  has slightly better regularity we may interpret the structural condition stated in \ref{aas:tc_sing_onedim} in a classical sense. To illustrate this, we define the linear differential operator
\begin{equation*}
\mathcal{Q}(t) \df \frac{1}{2}\sum_{j,k=1}^2A^{jk}(t,x)\frac{\partial^2 }{\partial x^j\partial x^k}
+ \sum_{j=1}^2 B^j(t,x) \frac{\partial }{\partial x^j} - C(t,x), \quad t\in [0,1],
\end{equation*}
and assume that $g=g(x)$ is twice weakly differentiable. Then $\mathcal{B}_{K}[g,\varphi;1]\neq 0$ for all bounded, open sets $K$ in $\R^2$ is equivalent to the assumption that $\mathcal{Q}(1)g\neq 0$ almost everywhere on $\R^2$.
\end{remark}



The proof of Theorem \ref{thm:FBMR} is given in Section \ref{sec:proof} and relies on specific smoothness and integrability properties of the solution to a parabolic equation,
which we obtain in Section \ref{sec:PDE} and on the invertibility of a
Jacobian matrix, which we study in Section \ref{sec:Jacobian}.

\section{Regularity of the Solution to the Associated Parabolic Equation}\label{sec:PDE}
For $(t,x)\in[0,1]\times\R^2$, consider an elliptic operator
\begin{equation}\label{eq:PDE_operator_ND}
\mathcal{G}(t) \df \sum_{j,k=1}^2 a^{jk}(t,x)\frac{\partial^2}{\partial
  x^j \partial x^k} + \sum_{j=1}^2 b^j(t,x)\frac{\partial}{\partial
  x^j} + c(t,x), 
\end{equation}
where the coefficients $a^{jk}, b^j, c : [0,1]\times\mathbb{R}^2 \to \mathbb{R}$ are measurable
functions and satisfy:

\begin{myenumerateB}
\item\label{as:PDE_1}
The maps $t\mapsto a^{jk}(t,\cdot)$, $t\mapsto b^j(t,\cdot)$, $t\mapsto
  c(t,\cdot)$ of $[0,1]$ to $\mathbf{C}$ are H\"{o}lder continuous and their restriction to $(0,1)$ is analytic. The map $t\mapsto a^{jk}(t,\cdot)$ is continuous of $[0,1]$ to $\mathbf{C}^2$ and the maps $t\mapsto b^j(t,\cdot)$, $t\mapsto c(t,\cdot)$ are continuous of $[0,1]$ to $\mathbf{C}^1$. The matrix $a$ is symmetric: $a^{ij} =
  a^{ji}$, and uniformly elliptic: there exists $N>0$ such that
\begin{equation*}
ya(t,x)y\geq \frac{1}{N^2}|y|^2,\quad (t,x) \in [0,1]\times\R^2,\quad y\in\R^2
\end{equation*}
and the function $c$ is nonpositive:
\begin{equation*}
 c(t,x) \leq 0, \quad (t,x)\in [0,1]\times\R^2. 
\end{equation*}
\end{myenumerateB}


Let $g=g(x) : \R^2 \to \R$ be a measurable function such that for some
$p>1$:
\begin{myenumerateB}
\item \label{as:PDE_2} the function $g$ belongs to $\mathbf{W}^1_p$.
\end{myenumerateB}

\begin{theorem}\label{thm:PDE_1}
Suppose that conditions \ref{as:PDE_1} and
\ref{as:PDE_2} hold. Then there exists a unique measurable function
$v=v(t,x)$ on $[0,1]\times\R^2$ such that
\begin{enumerate}[leftmargin=*,label={\arabic{enumi}.}]
\item $t\mapsto v(t,\cdot)$ is a continuous map of $[0,1]$ to
  $\mathbf{W}^1_p$,\\
\item $t\mapsto v(t,\cdot)$ is an analytic map of $(0,1)$ to $\mathbf{W}^2_p$,\\
\item $t\mapsto v(t,\cdot)$ is a $p$-integrable map of $[0,1)$ to $\mathbf{W}^3_p$,\\
\item $t\mapsto \partial_t v(t,\cdot)$ is a $p$-integrable map of $[0,1)$ to $\mathbf{W}^1_p$
\end{enumerate}
and such that $v=v(t,x)$ solves the homogeneous Cauchy problem
\begin{align}
 \left(\frac{\partial }{\partial t} + \mathcal{G}(t)\right)v &= 0, \quad t\in[0,1),\label{eq:cp_PDE}\\
 v(1,\cdot) &= g.\label{eq:cp_PDE_tc}
\end{align}
\end{theorem}

\begin{proof}
By assumption \ref{as:PDE_1} we know that for each $t\in[0,1]$ and $j,k=1,2$, the function $a^{jk}(t,\cdot)$ is in $\mathbf{C}^2$. In particular, the first-order partial derivatives of $a^{jk}$ with respect to $x$ are bounded and therefore the matrix $a$ is uniformly continuous with respect to $x$. Under the assumptions \ref{as:PDE_1} and \ref{as:PDE_2} the assertions of items one and two are immediately obtained upon making the time change $t\to 1-t$ in Theorem 3.1 in \cite{dKramkov2012}.

In addition Theorem 3.1 in \cite{dKramkov2012} tells us that $t\mapsto v(t,\cdot)$ is a continuously differentiable map of $[0,1)$ to $\mathbf{L}_p$ and a continuous map of $[0,1)$ to $\mathbf{W}^2_p$, which implies that $v=v(t,x)$ belongs to $\mathbf{W}^{1,2}_p([0,1)\times\R^2)$. Therefore, given the symmetry and the uniform ellipticity of the matrix-function $a=a(t,x)$, the uniform continuity of $a(t,\cdot)$, the fact that each function $a^{jk}(t,\cdot)$, $b^j(t,\cdot)$, $c(t,\cdot)$ belongs to $\mathbf{C}^1$ and the nonnegativity of $c=c(t,x)$, we may use Corollary 5.2.4 in \cite{nKrylov2008} to deduce that $v=v(t,x)$ in fact belongs to $\mathbf{W}^{1,3}_p([0,1)\times\R^2)$. The regularity in items three and four follows immediately.
\end{proof}

%
In the next section we will require the following corollary of Theorem \ref{thm:PDE_1}, where instead of \ref{as:PDE_2} we assume that the measurable function $g=g(x)$ is once weakly differentiable and has the following
property:

\begin{myenumerateB}
\item\label{as:ic_weighted} There exists a constant $N\geq 0$ such
  that
\begin{equation*}
e^{-N|\cdot|}\frac{\partial g}{\partial x^j} (\cdot) \ \in
\mathbf{L}_\infty,\quad j=1,2.
\end{equation*}
\end{myenumerateB}

Fix a function $\phi = \phi(x) : \R^2 \to \R$, which satisfies
\begin{equation}\label{eq:test_function}
\phi \in \mathbf{C}^\infty(\R^2) \text{ and } \phi(x) = |x| \text{ when } |x|\geq 1.
\end{equation}

\begin{corollary}\label{cor:PDE_2}
Suppose conditions \ref{as:PDE_1} and \ref{as:ic_weighted} hold. Let
$\phi=\phi(x)$ satisfy condition \eqref{eq:test_function}. Then
there exists a unique continuous function $v=v(t,x)$ on
$[0,1]\times\R^2$ and a constant $N\geq 0$ such that for every $p\geq
1$
\begin{enumerate}[leftmargin=*,label={\arabic{enumi}.}]
\item $t\mapsto e^{-N\phi(\cdot)}v(t,\cdot)$ is a continuous map of $[0,1]$ to
  $\mathbf{W}^1_p$,\\
\item $t\mapsto e^{-N\phi(\cdot)}v(t,\cdot)$ is an analytic map of $(0,1)$ to $\mathbf{W}^2_p$,\\
\item $t\mapsto e^{-N\phi(\cdot)}v(t,\cdot)$ is a $p$-integrable map of $[0,1)$ to $\mathbf{W}^3_p$,\\
\item $t\mapsto e^{-N\phi(\cdot)}\partial_tv(t,\cdot)$ is a $p$-integrable map of $[0,1)$ to $\mathbf{W}^1_p$
\end{enumerate}
and such that $v=v(t,x)$ solves the Cauchy problem \eqref{eq:cp_PDE} and \eqref{eq:cp_PDE_tc}.
\end{corollary}

\begin{proof}
From Assumption \ref{as:ic_weighted} we deduce the existence of a
constant $M>0$ such that
\begin{equation*}
\left|\frac{\partial g}{\partial x^i}(x)\right| \leq Me^{M|x|}, \quad x\in\R^2, 
\end{equation*}
and, therefore, such that
\begin{equation*}
\left|g(x)\right| \leq |x| Me^{M|x|} + M, \quad x\in\R^2.
\end{equation*}
One easily verifies now that for $N>M$ and $\phi=\phi(x)$ satisfying \eqref{eq:test_function} $\left\|e^{-N\phi} g\right\|_{\mathbf{W}^1_p}<\infty
$ for every $p\geq 1$ and hence that
\begin{equation}\label{eq:tc_Sobolev}
e^{-N\phi} g \in \mathbf{W}^1_p,\quad
p\geq 1.
\end{equation}
Hereafter we choose the constant $N\geq 0$ from \ref{as:ic_weighted}
to also satisfy $N>M$.

Let $C\geq 0$ be a constant and define the functions $\tilde{b}^j = \tilde{b}^j(t,x)$ and
$\tilde{c}=\tilde{c}(t,x)$, so that for $t\in[0,1]$ and
$u\in\mathbf{C}^\infty((0,1)\times\R^2)$,
\begin{equation*}
\left(\frac{\partial }{\partial t} + \tilde{\mathcal{G}}(t)\right)(e^{-N\phi + Ct}u) = e^{-N\phi+Ct}\left(\frac{\partial }{\partial t}+\mathcal{G}(t)\right)u,
\end{equation*}
where
\begin{equation*}
\tilde{\mathcal{G}}(t) \df \sum_{j,k=1}^2 a^{jk}(t,x)\frac{\partial^2}{\partial
  x^j \partial x^k} + \sum_{j=1}^2 \tilde{b}^j(t,x)\frac{\partial}{\partial
  x^j} + \tilde{c}(t,x). 
\end{equation*}
Given the properties of the function $\phi$ asserted in
\eqref{eq:test_function}, for $C$ large enough, the coefficients
$\tilde{b}^j$ and $\tilde{c}^j$ satisfy the same conditions as $b^j$ and
$c$ in \ref{as:PDE_1}. Since it follows from \eqref{eq:tc_Sobolev} that also $e^{-N\phi+C}g$ belongs to $\mathbf{W}^1_p$, for every $p\geq 1$, we deduce from Theorem \ref{thm:PDE_1} the existence of
a measurable function $\tilde{v} = \tilde{v}(t,x)$, which, for every $p>1$,
complies with items one to four of Theorem \ref{thm:PDE_1} and solves the
Cauchy problem
\begin{align}
 \frac{\partial \tilde{v}}{\partial t} + \tilde{\mathcal{G}}(t) \tilde{v} &= 0, \quad t\in [0,1),\label{eq:PDE_cor}\\
 \tilde{v}(1,\cdot) &= e^{-N\phi+C}g.\label{eq:PDE_cor_tc}
\end{align}
For $p>2$, by Sobolev's embedding theorem, the continuity of the map
$t\mapsto \tilde{v}(t,\cdot)$ in $\mathbf{W}^1_p$ implies its
continuity in $\mathbf{C}$. It follows that the function
$\tilde{v}=\tilde{v}(t,x)$ is continuous on $[0,1]\times\R^2$.

Defining $v\df e^{N\phi-Ct}\tilde{v}$, we observe that $\tilde{v}$ solves
\eqref{eq:PDE_cor} and \eqref{eq:PDE_cor_tc}, if and only if $v$ solves the
Cauchy problem \eqref{eq:cp_PDE} and \eqref{eq:cp_PDE_tc}. For $p>1$, the regularity of $\tilde{v} = e^{-N\phi+Ct}v$ implies items one to four in the corollary. The
proof is completed by noting that the case $p=1$ follows trivially
from the case $p>1$ by taking the constant $N$ slightly larger.
\end{proof}

For $(t,x)\in [0,1]\times\R^2$ define
\begin{equation}\label{eq:partial}
v_j(t,x) \df \frac{\partial v}{\partial x^j}(t,x),\quad j = 1,2,
\end{equation}
and consider the elliptic operator
\begin{equation}\label{eq:diff_ell_op}
\mathcal{G}_l(t) \df \sum_{j,k=1}^2 \frac{\partial a^{jk}}{\partial x^l}(t,x)\frac{\partial^2}{\partial
  x^j \partial x^k} + \sum_{j=1}^2 \frac{\partial b^j}{\partial
  x^l}(t,x)\frac{\partial}{\partial
  x^j} + \frac{\partial c}{\partial x^l}(t,x),\quad l=1,2.
\end{equation}

Then we obtain the following corollary, which will be needed in the next section.
\begin{corollary}\label{cor:PDE_deriv}
Suppose that conditions \ref{as:PDE_1} and
\ref{as:ic_weighted} hold. Let $v=v(t,x)$ be the function generated by
Corollary \ref{cor:PDE_2} and let $v_j$ be defined as in
\eqref{eq:partial}. Then $v_j = v_j(t,x)$ solves the nonhomogeneous partial differential equation
\begin{equation}\label{eq:deriv_PDE}
 \frac{\partial v_j}{\partial t} + \mathcal{G}(t)v_j + \mathcal{G}_j(t)v = 0, \quad t\in(0,1).
\end{equation}
\end{corollary}

\begin{proof}
From Corollary \ref{cor:PDE_2} we know that the function $v=v(t,x)$ is three times weakly differentiable with respect to $x$ and that the derivative with respect to $t$ of the same function is once weakly differentiable with respect to $x$. Given condition \ref{as:PDE_1} we also know that the coefficients of the operator $\mathcal{G}$ are once continuously differentiable with respect to $x$. Hence we may differentiate the parabolic partial differential equation \eqref{eq:cp_PDE} with respect to $x^j$, $j=1,2$, which shows that $v_j = v_j(t,x)$ satisfies \eqref{eq:deriv_PDE}.
\end{proof}

\section{Invertibility of the Jacobian Matrix}\label{sec:Jacobian}
Let $a=a(t,x)$, $b=b(t,x)$, $c=c(t,x)$ and $g=g(x)$ be the coefficients from Section \ref{sec:PDE}. Let the measurable function $f=f(t,x): [0,1]\times\R^2\to \R$ be three times weakly differentiable with respect to $x$ and assume that there exists a constant $N\geq 0$ such that, for $j,k,l=1,2$, it holds that:
\begin{myenumerateB}
\item\label{aaas:forw_f_generic} 
The map $t\mapsto e^{-N|\cdot|}\partial_{x^jx^k}f(t,\cdot)$ of $(0,1)$ to $\mathbf{L}_\infty$ is analytic, the map $t\mapsto e^{-N|\cdot|}\partial_{x^j}f(t,\cdot)$ of $[0,1]$ to $\mathbf{L}_\infty$ is continuously differentiable and the map $t\mapsto e^{-N|\cdot|}\partial_{x^jx^kx^l}f(t,\cdot)$ of $[0,1]$ to $\mathbf{L}_\infty$ is continuous.
\end{myenumerateB}

We define the functions $A=A(t,x)$, $B=B(t,x)$ and $C = C(t,x)$ on $[0,1]\times\R^2$ by 
\begin{equation*}
\begin{aligned}
 A^{jk} &\df |J[f,a^{jk}]|-2(-1)^j(H[f]a)^{(3-j)k},\\
 B^j &\df |J[f,b^j]| -(-1)^j(\partial_t + \mathcal{G}(t))\partial_{x^{(3-j)}}f,\\
  C &\df |J[f,c]|,
 \end{aligned}
\end{equation*}
for $j,k=1,2$.

For suitably regular functions $v,\varphi :\R^2\to\R$, for an open, bounded set $K$ in $\R^2$ and for $t\in [0,1]$, we define the pairing
\begin{equation*}
\mathcal{A}_K[v,\varphi;t] \df \int_K \sum_{j,k=1}^2A^{jk}(t,x)\frac{\partial v}{\partial x^j}\frac{\partial \varphi}{\partial x^k}
- \sum_{j=1}^2\left(B^j - \sum_{k=1}^2\frac{\partial A^{jk}}{\partial x^k}\right)(t,x)\frac{\partial v}{\partial x^j}\varphi - C(t,x)v\varphi\ \Id x.
\end{equation*}

We assume that the following assumption is satisfied:
\begin{myenumerateB}
\item\label{as:tc_sing_onedim_str} Either the Jacobian matrix $J[f,g](1,\cdot)$ has full rank almost everywhere on $\R^2$ or for every open, bounded set $K$ in $\R^2$ there exists a test function $\varphi=\varphi(x)$ belonging to $\mathbf{W}^1_{p,0}(K)$, for some $p\geq1$, such that $\mathcal{A}_K[g,\varphi;1]\neq 0$.
\end{myenumerateB}

The following theorem is the main result of this section and will
eventually allow us to prove the martingale representation stated in
Theorem \ref{thm:FBMR}.
\begin{theorem}\label{thm:rank}
Suppose conditions \ref{as:PDE_1}, \ref{as:ic_weighted}, \ref{aaas:forw_f_generic} and \ref{as:tc_sing_onedim_str} are in place. Let $v=v(t,x)$ be the function furnished by Corollary \ref{cor:PDE_2}. Then 
the Jacobian matrix-function $J[f,v] = J[f,v](t,x)$ has full rank almost everywhere with respect to the Lebesgue measure on $[0,1]\times\R^2$.
\end{theorem}

Before we can proof Theorem \ref{thm:rank} we first need to establish several lemmas below.

Let $\mathbf{X}$ and $\mathbf{Y}$ be Banach spaces, $\mathbf{E}$ an
open subset of $\mathbf{X}$ and consider a map
$h:\mathbf{E}\to\mathbf{Y}$. If it exists, we denote by $D^khx$ the $k$-th Fr\'{e}chet
derivative of $h$ at the point $x\in \mathbf{E}$; as is well known,
this constitutes a $k$-linear map on
the $k$-fold product $\mathbf{X}\times\ldots\times\mathbf{X}$. Accordingly, for
$x^1,\ldots,x^k\in\mathbf{X}$ we denote by $D^khx(x^1,\ldots,x^k)$ the $k$-th Fr\'{e}chet differential.

\begin{lemma}\label{lem:frechet_derivative}
Given matrices $M,C,C^1,C^2\in\mathbb{R}^{2\times 2}$, the first and second order Fr\'{e}chet differentials of the determinant map at $M$ are given by
\begin{align*}
D \det M(C) &= \sum_{l=1}^2\det M(l;C),\\
D^2 \det M(C^1,C^2) &= \sum_{l=1}^2\det M(1,2;C^l,C^{3-l}).
\end{align*}
\end{lemma}
\begin{proof}
The expressions are special cases of equations (4) and (6) in \cite{rBhatia2009}.
\end{proof}

Define the linear partial differential operator
\begin{equation*}
 \mathcal{P}(t) \df \sum_{j,k=1}^2 A^{jk}(t,x)\frac{\partial^2}{\partial x^j \partial x^k} + \sum_{j=1}^2 B^j(t,x) \frac{\partial }{\partial x^j} + C(t,x), \quad t\in [0,1].
\end{equation*}

\begin{lemma}\label{lem:PDE_det}
Let $f =f(t,x),v =v(t,x) : [0,1]\times\R^2 \to \R$ be measurable functions, which on $(0,1)\times\R^2$ are once weakly differentiable with respect to $t$, three times weakly differentiable with respect to $x$ and once weakly differentiable with respect to $t$ and $x$, and let $\mathcal{G}(t)$ and $\mathcal{G}_j(t)$ be the operators defined in \eqref{eq:PDE_operator_ND} and \eqref{eq:diff_ell_op} respectively. Define $f_j \df \partial_{x^j}f$, $v_j \df \partial_{x^j}v$, $j=1,2$, and assume that $v_j$ satisfies the partial differential equation
\begin{equation}
 \frac{\partial v_j}{\partial t} + \mathcal{G}(t)v_j + \mathcal{G}_j(t)v = 0, \quad t\in(0,1).\label{eq:lem_PDE}
\end{equation}
Then the determinant function $w = w(t,x)$ defined on $[0,1]\times\R^2$ by $w \df |J[f,v]|$
%
satisfies the nonhomogeneous partial differential equation
\begin{equation}\label{eq:some_cond}
 \frac{\partial w}{\partial t} + (\mathcal{G}(t)+c) w = -\mathcal{P}(t)v.
\end{equation}
\end{lemma}

\begin{proof}
Given our differentiability hypothesis on $f=f(t,x)$ and $v=v(t,x)$ we may differentiate the determinant function $w=w(t,x)$ with respect to $t$. Let us abbreviate throughout the proof of this lemma $J\df J[f,v]$. A simple application of the chain rule from Fr\'{e}chet differential calculus (cf. \cite[Chapter X.4]{rBhatia1997}) and the fact that $v_j$ satisfies the partial differential equation \eqref{eq:lem_PDE} yields
\begin{equation}\label{eq:det_time_evol}
\frac{\partial w}{\partial t} =
D \det J \left(
-\mathcal{G}(t)J
-\left(
\begin{array}{c}
0\\
\nabla_x\mathcal{G}(t)
\end{array}
\right)
v
+\left(\frac{\partial }{\partial t}+\mathcal{G}(t)\right)
\left(
\begin{array}{c}
\nabla_x f\\
0
\end{array}
\right)
\right) ,\quad t\in (0,1).
\end{equation}

The direct computation of $\mathcal{G}(t)w$ and making use of the identity $2c\det J = c D \det J(J)$ and the linearity of the Fr\'{e}chet derivative show that we may replace the term $-D\det J (\mathcal{G}(t)J)$ above with
\begin{equation*}
\sum_{j,k=1}^2 a^{jk} D^2 \det J \left(\frac{\partial J}{\partial x^j}, \frac{\partial J}{\partial x^k} \right)- (\mathcal{G}(t) + c) w, \quad t\in(0,1).
\end{equation*}
By the explicit formulae for the first and second order Fr\'{e}chet derivative of the determinant map derived in Lemma \ref{lem:frechet_derivative} and the symmetry of the matrix-function $a$, we obtain after some computations
\begin{equation*}
 \frac{\partial w}{\partial t} + (\mathcal{G}(t)+c) w =
2\sum_{j,k=1}^2 a^{jk} |J[f_j,v_k]|
+\left|
 \begin{array}{c}
(\partial_{t}+\mathcal{G}(t))\nabla_x f\\
\nabla_x v
\end{array}
 \right|
 -\left|
 \begin{array}{c}
\nabla_x f\\
(\nabla_x \mathcal{G}(t)) v
 \end{array}
 \right|.
 \end{equation*}
Collecting the coefficients of $\partial^2_{x^jx^k}v$, $\partial_{x^j}v$ and $v$ yields the result.
\end{proof}

\begin{lemma}\label{lem:dual_cont}
 Let $\gamma^j,\eta : [0,1]\times\R^2 \to \R$, $j=1,2$, be measurable functions such that, for $p>1$, the maps $t\mapsto\gamma^j(t,\cdot),t\mapsto\eta(t,\cdot)$ of $[0,1]$ to $\mathbf{L}_{p,\text{loc}}$ are continuous. Let $K$ be an open, bounded set in $\R^2$ and $\varphi=\varphi(x)$ a test function belonging to $\mathbf{W}^1_{p',0}(K)$. Then, for each $t\in[0,1]$, the pairing
 \begin{equation*}
 \tilde{\mathcal{A}}_K(\varphi;t) \df \int_K \sum_{j=1}^2 \gamma^j(t,x) \frac{\partial \varphi}{\partial x^j} + \eta(t,x)\varphi\ \Id x
 \end{equation*}
is a bounded, linear functional on $\mathbf{W}^1_{p',0}(K)$. Moreover, the map $t\mapsto \tilde{\mathcal{A}}_K(\cdot;t)$ is continuous as a map of $[0,1]$ to $\mathbf{W}^{-1}_p(K)$.
\end{lemma}

\begin{proof}
By the triangle inequality and the H\"{o}lder inequality, for each $t\in[0,1]$,
\begin{equation*}
 \begin{aligned}
\left| \tilde{\mathcal{A}}_K(\varphi;t)\right| &\leq \int_K \sum_{j=1}^2 \left|\gamma^j(t,x) \frac{\partial \varphi}{\partial x^j}\right| + |\eta(t,x)\varphi|\ \Id x\\
&\leq \| \varphi\|_{\mathbf{W}^1_{p'}(K)}\left(\sum_{j=1}^2 \|\gamma^j(t,\cdot) \|_{\mathbf{L}_p(K)} + \|\eta(t,\cdot)\|_{\mathbf{L}_p(K)}\right),
\end{aligned}
\end{equation*}
which implies the boundedness of the linear functional $\tilde{\mathcal{A}}_K(\cdot;t)$.

To prove the continuity of the map $t\mapsto \tilde{\mathcal{A}}_K(\cdot;t)$ of $[0,1]$ to $\mathbf{W}^{-1}_{p}(K)$, observe that, for each $t\in[0,1]$, $\tilde{\mathcal{A}}_K(\cdot;t)$ is in the dual space of $\mathbf{W}^1_{p',0}(K)$. We recall that the dual space of $\mathbf{W}^1_{p',0}(K)$ is isometrically isomorphic to $\mathbf{W}^{-1}_p(K)$. It follows that
\begin{equation*}
\begin{aligned}
 &\left\| \tilde{\mathcal{A}}_K(\cdot;t) - \tilde{\mathcal{A}}_K(\cdot;u)\right\|_{\mathbf{W}^{-1}_p(K)}\\
 \leq &\sum_{j=1}^2 \|\gamma^j(t,\cdot) - \gamma^j(u,\cdot)\|_{\mathbf{L}_p(K)}
 + \|\eta(t,\cdot) - \eta(u,\cdot)\|_{\mathbf{L}_p(K)},
 \end{aligned}
\end{equation*}
which implies the desired continuity of the map $t\mapsto \tilde{\mathcal{A}}_K(\cdot;t)$ by the continuity of the maps $t\mapsto\gamma^j(t,\cdot),t\mapsto\eta(t,\cdot)$ of $[0,1]$ to $\mathbf{L}_{p,\text{loc}}$.
\end{proof}


\begin{proof}
For the proof of Theorem \ref{thm:rank} we define
\begin{equation*}
 w(t,x) \df |J[f,v]|(t,x), \qquad (t,x)\in [0,1]\times\R^2.
\end{equation*}
The claim of the theorem is true if and only if the set
\begin{equation*}
 G \df\{(t,x)\in [0,1]\times\R^2 : w(t,x) = 0\}
\end{equation*}
has Lebesgue measure zero on $[0,1]\times\R^2$. This is equivalent to the set
\begin{equation*}
 H \df\{x\in \R^2 : \int_0^1 1_G(t,x)\ \Id t > 0\}
\end{equation*}
having Lebesgue measure zero on $\R^2$.

From Corollary \ref{cor:PDE_2} and \ref{aaas:forw_f_generic} we deduce that, for every $p\geq 1$, the map $t\mapsto e^{-N\phi(\cdot)}w(t,\cdot)$ is analytic as a map of $(0,1)$ to $\mathbf{W}^1_p$. Moreover, by Sobolev's embedding theorems, for $p>2$, it is also analytic as a map of $(0,1)$ to $\mathbf{C}$. Suppose, for a contradiction, that
\begin{equation*}
\int_{\R^2} 1_H(t,x)\ \Id x >0.
\end{equation*}
From the analyticity of $t\mapsto w(t,\cdot)$, it follows that if $x\in H$ then $w(t,x)=0$ for all $t\in (0,1)$ and therefore that
\begin{equation*}
 \lim_{t\uparrow 1}w(t,x) = 0, \qquad x\in H.
\end{equation*}

We first proof the claim of the theorem assuming $J[f,g](1,x)$ has full rank almost everywhere on $\R^2$. From Corollary \ref{cor:PDE_2} and \ref{aaas:forw_f_generic} we know that, for every $p\geq 1$, the map $t\mapsto
e^{-N\phi(\cdot)}w(t,\cdot)$ of $[0,1]$ to $\mathbf{L}_p$ is
continuous. It follows that the map $t\mapsto w(t,\cdot)$ of $[0,1]$ to $\mathbf{L}_{p,\text{loc}}$ is continuous and hence that, for all open, bounded sets $K$ in $\R^2$,
\begin{equation*}
 \|w(t,\cdot) - w(1,\cdot)\|_{\mathbf{L}_{p}(K)} \to 0, \qquad t\uparrow 1.
\end{equation*}
We deduce that $w(1,\cdot) = |J[f,v]|(1,\cdot) = 0$ almost everywhere. Now recall that
by \ref{as:tc_sing_onedim_str} the matrix-function $J[f,v](1,\cdot) = J[f,g](1,\cdot)$ has full rank almost everywhere on $\R^2$.

%
%
%
Let us now assume that for every open, bounded set $K$ in $\R^2$ there exists a test function $\varphi = \varphi(x)$ belonging to $\mathbf{W}^1_{p',0}(K)$ such that $\mathcal{A}_K[g,\varphi;1]\neq 0$. From Corollary \ref{cor:PDE_2} and \ref{aaas:forw_f_generic} we know that the functions $f=f(t,x)$ and $v=v(t,x)$ satisfy the differentiability hypothesis of Lemma \ref{lem:PDE_det} and from Corollary \ref{cor:PDE_deriv} that $v_j$ satisfies the partial differential equation \eqref{eq:lem_PDE}. It follows from \eqref{eq:some_cond} that if $w(t,x) = 0$ for all $(t,x)\in (0,1)\times H$, then also $\mathcal{P}(t)v=0$ for all $(t,x)\in (0,1)\times H$.

From Corollary \ref{cor:PDE_2} we know that, for every $p\geq 1$, $t\mapsto e^{-N\phi(\cdot)}v(t,\cdot)$ is a continuous map of $[0,1]$ to $\mathbf{W}^1_p$. In particular, for every $p> 1$, $t\mapsto v(t,\cdot)$ and $t\mapsto \partial_{x^j}v(t,\cdot)$ are continuous maps of $[0,1]$ to $\mathbf{L}_{p,\text{loc}}$. Assumptions \ref{as:PDE_1} and \ref{aaas:forw_f_generic} imply that also $t\mapsto A^{jk}(t,\cdot), t\mapsto B^j(t,\cdot), t\mapsto C(t,\cdot), t\mapsto \partial_{x^k}A^{jk}(t,\cdot)$ are continuous maps of $[0,1]$ to $\mathbf{L}_{p,\text{loc}}$, for every $p> 1$. It follows from Lemma \ref{lem:dual_cont} that for any open, bounded set $K^{'}\subset H$, for any test function $\varphi = \varphi(x)$ of class $\mathbf{W}^1_{p^{'},0}(K^{'})$ and for any fixed $t\in[0,1]$ the pairing $\mathcal{A}_{K^{'}}[v,\cdot;t]$ is a bounded, linear functional on
$\mathbf{W}^1_{p^{'},0}(K^{'})$. Actually, for $t\in(0,1)$,
$\mathcal{A}_{K^{'}}[v,\varphi;t] = 0$ is the weak formulation of the
partial differential equation $\mathcal{P}(t)v = 0$. It follows that
$\mathcal{A}_{K^{'}}[v,\varphi;t] = 0$ for all $t \in (0,1)$ and every
$\varphi=\varphi(x)$ belonging to $\mathbf{W}^1_{p^{'},0}(K^{'})$ and therefore that
\begin{equation*}
 \lim_{t\uparrow 1}\mathcal{A}_{K^{'}}[v,\varphi;t] = 0, \qquad \varphi \in \mathbf{W}^1_{p^{'},0}(K^{'}).
\end{equation*}

Also from Lemma \ref{lem:dual_cont} we know that the map $t\mapsto \mathcal{A}_{K^{'}}[v,\cdot;t]$ is continuous as a map of $[0,1]$ to $\mathbf{W}^{-1}_{p}(K^{'})$ and therefore that
\begin{equation*}
\| \mathcal{A}_{K^{'}}[v,\cdot;t] - \mathcal{A}_{K^{'}}[v,\cdot;1]
\|_{\mathbf{W}^{-1}_p(K^{'})} \to 0, \qquad t\uparrow 1.
\end{equation*}
It follows that $\mathcal{A}_{K^{'}}[v,\varphi;1] = \mathcal{A}_{K^{'}}[g,\varphi;1] =0$ for every $\varphi \in
\mathbf{W}^1_{p^{'},0}(K^{'})$. Now recall that by \ref{as:tc_sing_onedim_str} for every open, bounded set $K$ and some $p>1$ there exists a test function $\varphi \in
\mathbf{W}^1_{p^{'},0}(K)$ such that $\mathcal{A}_K[g,\varphi;1] \neq 0$.
\end{proof}

\section{Proof of Theorem \ref{thm:FBMR}}\label{sec:proof}
From here onwards we adopt the notation introduced in Section \ref{sec:MR} and assume that conditions \ref{as:drift_vol_coeffs}, \ref{as:tc_dc_coeffs}, \ref{as:A3} and \ref{aas:tc_sing_onedim} are in place.

We fix a function $\phi=\phi(x)$ on $\R^2$ satisfying $\eqref{eq:test_function}$ and recall that $\mathcal{L}^X(t)$, $t\in[0,1]$, is the infinitesimal generator of the process $X$:

\begin{lemma}\label{lem:full_rank}
There exists a unique continuous function $v=v(t,x)$, on
$[0,1]\times\R^2$ and a constant $N\geq 0$ such that the following hold:
\begin{enumerate}[leftmargin=*,label={\arabic{enumi}.}]
\item\label{thm:rank_reg} For every $p\geq 1$, 
\begin{enumerate}[leftmargin=*,label={(\alph{enumii}})]
\item $t\mapsto e^{-N\phi(\cdot)}v(t,\cdot)$ is a continuous map of $[0,1]$ to $\mathbf{W}^1_p$,\\
\item $t\mapsto e^{-N\phi(\cdot)}v(t,\cdot)$ is an analytic map of $(0,1)$ to $\mathbf{W}^2_p$,\\
\item $t\mapsto e^{-N\phi(\cdot)}v(t,\cdot)$ is a $p$-integrable map of $[0,1)$ to $\mathbf{W}^3_p$,\\
\item $t\mapsto e^{-N\phi(\cdot)}\partial_tv(t,\cdot)$ is a $p$-integrable map of $[0,1)$ to $\mathbf{W}^1_p$,\\
\end{enumerate}
\item The function $v=v(t,x)$ solves the homogeneous Cauchy problem
\begin{align}
  \frac{\partial v}{\partial t} + (\mathcal{L}^X(t) -r)v &= 0, \quad t\in[0,1),\label{eq:cp}\\
 v(1,\cdot) &= g.\label{eq:cp_tc}
\end{align}
\item \label{lem:full_rank_jac} The Jacobian matrix-function $J[f,v]=J[f,v](t,x)$,
has full rank almost everywhere with respect to the Lebesgue measure on $[0,1]\times\R^2$.
\end{enumerate}
\end{lemma}


Hereafter we denote by $v=v(t,x)$, the function defined in Lemma \ref{lem:full_rank}.

\begin{proof}
 Observe that \ref{as:drift_vol_coeffs}, \ref{as:tc_dc_coeffs}, \ref{as:A3} and
 \ref{aas:tc_sing_onedim} imply \ref{as:PDE_1}, \ref{as:ic_weighted}, \ref{aaas:forw_f_generic} and\ref{as:tc_sing_onedim_str} on the
 corresponding coefficients in Theorem \ref{thm:rank}. The assertions for $v$ and $J[f,v]$ now follow directly from Theorem \ref{thm:rank}.
\end{proof}

\begin{lemma}\label{lem:backward_construction}
The martingale
\begin{equation*}
 S^B_t \df \E[\psi|\mathcal{F}_t],
\end{equation*}
is well-defined and has the representation
\begin{equation}\label{eq:backward_asset_defn}
 S^B_t = v(t,X_t)e^{-\int_0^t r(u,X_u)\ \Id u}.
\end{equation}
Moreover, for $t\in (0,1)$,
\begin{equation}\label{eq:cond_exp_evol}
 \Id S^B_t = e^{-\int_0^t r(u,X_u)\ \Id
   u}(\nabla_x v\sigma)(t,X_t)\ \Id W_t.
\end{equation}
\end{lemma}

\begin{proof}
%
Assume that the process $S^B$ is actually \textit{defined} by
\eqref{eq:backward_asset_defn}. From the continuity of $v$ on
$[0,1]\times\R^2$ it follows that it is in fact a continuous
process on $[0,1]$ and from the expression \eqref{eq:cp_tc} for
$v(1,\cdot)$ that $S^B_1=\psi$. Hence, to complete the proof it
remains to show that $S^B$, given by
\eqref{eq:backward_asset_defn} is a martingale under the measure $\mathbb{P}$.

From Lemma \ref{lem:full_rank} we know that the map $t\mapsto e^{-N\phi(\cdot)}v(t,\cdot)$ is analytic as a map of $(0,1)$ to $\mathbf{W}^2_p$; in particular, it is continuously differentiable. This allows us to use a variant of the It\^{o} formula due to Krylov (cf. \cite[Section 2.10, Theorem 1]{nKrylov1980}) and accounting for \eqref{eq:cp}, we immediately obtain
\eqref{eq:cond_exp_evol}.

We have shown that $S^B$ is a continuous local martingale. It
only remains to verify the uniform integrability of the
process. Recall that, for every $p\geq1$, the map $t\mapsto
e^{-N\phi(\cdot)}v(t,\cdot)$ is continuous from $[0,1]$ to
$\mathbf{W}^1_p$. It follows from Sobolev's embedding theorem that,
for $p>2$, the same map is continuous from $[0,1]$ to
$\mathbf{C}$. Therefore,
\begin{equation*}
|v(t,x)| \leq e^{N(1+|x|)}.
\end{equation*}
In particular, accounting for the growth properties of $r=r(t,x)$,
\begin{equation*}
\sup_{t\in[0,1]}(|S^B_t|) \leq e^{N(1+\sup_{t\in [0,1]}|X_t|)}.
\end{equation*}
As $\sup_{t\in[0,1]}|X_t|$ has all exponential moments the martingale
property for $S^B$ follows.
\end{proof}

The proof of Theorem \ref{thm:FBMR} is now completed
easily. Equations \eqref{eq:FB_asset_evol} and \eqref{eq:cond_exp_evol} show that
\begin{equation}\label{eq:mart}
\left\{
\begin{aligned}
 \Id S^F_t &= e^{-\int_0^t r(u,X_u)\ \Id u}(\nabla_x f\sigma)(t,X_t)\ \Id W_t,\\
 \Id S^B_t &= e^{-\int_0^t r(u,X_u)\ \Id
   u}(\nabla_x v\sigma)(t,X_t)\ \Id W_t.
\end{aligned}
\right.
\end{equation}
By the growth properties of $r = r(t,x)$, $f=f(t,x)$ and $\sigma=\sigma(t,x)$ in \ref{as:drift_vol_coeffs}, \ref{as:tc_dc_coeffs} and \ref{as:A3} it is easily verified that also the continuous local martingale $S^F = (S^F_t)$ is a true martingale.

In view of \eqref{eq:mart} we obtain
\begin{equation*}
\Id S_t = e^{-\int_0^t r(u,X_u)\ \Id u} ( J[f,v]\sigma)(t,X_t)\ \Id W_t, \quad t\in[0,1].
\end{equation*}
We recall that, by the Brownian integral representation property, every $\mathbb{P}$-local martingale $M$ is a stochastic integral with respect to $W$:
\begin{equation*}
 \Id M_t = \tilde{H}_t\ \Id W_t,\quad t\in[0,1],
\end{equation*}
for some progressively measurable, square-integrable process $\tilde{H} = (\tilde{H}_t)$. Hence, in order to deduce the integral representation property \eqref{eq:int_rep} it remains to show that the matrix process
\begin{equation}\label{eq:vol_S}
(J[f,v]\sigma)(t,X_t), \quad t\in [0,1],
\end{equation}
has full rank on $\Omega\times [0,1]$ almost surely under the product
measure $\Id \mathbb{P} \times \Id t$. From Lemma \ref{lem:full_rank} we know that the matrix-function $J[f,v] = J[f,v](t,x)$ has full rank almost everywhere under the Lebesgue measure on $[0,1]\times\R^2$. From the nonsingularity assumption in \ref{as:drift_vol_coeffs} we know that also the matrix-function $\sigma = \sigma(t,x)$ has full rank almost everywhere under the Lebesgue measure on $[0,1]\times\R^2$. The conclusion that \eqref{eq:vol_S} has full rank on $\Omega\times [0,1]$ almost surely now follows easily from the fact that under \ref{as:drift_vol_coeffs} the distribution of $X_t$ has a density under the Lebesgue measure on $\R^2$, see \cite[Theorem 9.1.9]{dStroock2006}.

\section{Example: a class of stochastic volatility models}\label{sec:example}
In this section we apply our main result Theorem \ref{thm:FBMR} to prove the completeness of a financial market in which one stock with price process $P=(P_t)$ and one call option with price process $V=(V_t)$ are traded. The processes $P$ and $V$ are defined by
\begin{equation}\label{eq:stoch_vol}
 \begin{aligned}
 \Id P_t &= rP_t\ \Id t + \nu(Y_t)P_t\ \Id W^1_t\\
 \Id Y_t &= (\alpha(m-Y_t) - \mu(P_t,Y_t))\ \Id t + \sigma(Y_t)\ \Id W_t\\
 V_t &= e^{-r(1-t)}\E[(P_1-\Gamma)^+|\mathcal{F}_t],
\end{aligned}
\end{equation}
for constants $\Gamma,\alpha,m,r\in\R$ with $\Gamma>0$, $r\geq 0$. In particular this covers the class of stochastic volatility models introduced in \cite[Equation (2.7), p.~43]{jFouque2000}.

The coefficients $\nu,\mu,\sigma^j:\R\to\R$, $j=1,2$, are assumed to satisfy the following condition:
\begin{myenumerateC}
 \item\label{as:stoch_vol} There exist constants $N,D,\rho,\epsilon>0$ such that for all $y\in \R$, $\nu(y)>N$ and $\sigma^j(y)>N$; the derivative $\Id \nu/\Id y(y) \neq 0$ almost everywhere on $\R$ and the functions $\nu$, $\sigma^j$ and $\mu(p,\cdot)$ are infinitely differentiable and satisfy
 \begin{equation*}
  \left|\frac{\partial^k\mu}{\partial y^k}(p,y)\right| + \left|\frac{\partial^k\nu}{\partial y^k}(y)\right| + \left|\frac{\partial^k\sigma^j}{\partial y^k}(y)\right|\leq \frac{Dk!}{(\rho+\epsilon|y|)^k}, \quad (p,y)\in \R\times\R.
 \end{equation*}
The function $\mu=\mu(p,y)$ has first and second continuous derivatives in $p$ and $y$ and $y(e^p)^l \partial_y^k\partial_p^l\mu\in\mathbf{L}_\infty$, $l=0,1$, $k=1,2$, $l+k\leq 2$.
\end{myenumerateC}

We are now ready to state the main result of this section.
\begin{theorem}\label{thm:stoch_vol}
Suppose that condition \ref{as:stoch_vol} is satisfied. Then the $(P_t,V_t)$-market defined by \eqref{eq:stoch_vol} is complete.
\end{theorem}

\begin{remark}
 We draw attention to the fact that in \ref{as:stoch_vol} the quite specific assumptions on the space regularity of the coefficients of \eqref{eq:stoch_vol} are solely necessary because we are allowing $P$ to evolve according to a geometric Brownian motion and $Y$ to have mean reverting dynamics, both cases in which the coefficients are unbounded.  This can be seen easily from the proof of Theorem \ref{thm:stoch_vol} below. In the absence of this particular choice of dynamics the verification of the assumptions of Theorem \ref{thm:FBMR} is much simpler.
\end{remark}

\begin{remark}
 Two specific examples of functions which satisfy the conditions on $\nu$ and $\sigma$ in \ref{as:stoch_vol} are scaled and shifted versions of the $arctan$ and $tanh$ functions. 
\end{remark}

\begin{proof}
Consider the stochastic processes
\begin{align*}
X^1_t &\df \log P_t, && S^F_t \df e^{-rt+X^1_t},\\
X^2_t&\df e^{\alpha t}(Y_t-m), &&S^B_t \df\E[e^{-r}(e^{X^1_1}-\Gamma)^+|\mathcal{F}_t].
\end{align*}
By a simple application of It\^{o}'s formula we find that
\begin{equation*}
\begin{aligned}
\Id X^1_t &= \left(r-\frac{1}{2}\nu(m+e^{-\alpha t}X^2_t)^2\right)\ \Id t + \nu(m+e^{-\alpha t}X^2_t)\ \Id W^1_t\\
\Id X^2_t &= -e^{\alpha t}\mu(e^{X^1_t},m+e^{-\alpha t}X^2_t)\ \Id t + e^{\alpha t}\sigma(m+e^{-\alpha t}X^2_t)\ \Id W_t.
\end{aligned}
\end{equation*}
We define $\tilde{\nu}(t,x^2)\df\nu(m+e^{-\alpha t}x^2)$, $\tilde{\sigma}^j(t,x^2)
\df\sigma^j(m+e^{-\alpha t}x^2)$ and $\tilde{\mu}(t,x^1,x^2)\df\mu(e^{x^1},m+e^{-\alpha t}x^2)$. Now observe that by Lemma \ref{lem:analytic_comp} in Appendix \ref{ap:analytic_comp} the maps $t\mapsto \tilde{\nu}(t,\cdot), \tilde{\sigma}^j(t,\cdot)$ of $[0,1]$ to $\mathbf{C(\R)}$ and the map $t\mapsto\tilde{\mu}(t,\cdot,\cdot)$ of $[0,1]$ to $\mathbf{C}$ are analytic. By computing the derivative with respect to $t$ and using the bounds on the derivatives of $\mu$, $\nu$ and $\sigma^j$ hypothesized in \ref{as:stoch_vol} it is verified easily that the maps $t\mapsto \tilde{\nu}(t,\cdot), \tilde{\sigma}(t,\cdot)$ are continuous of $[0,1]$ to $\mathbf{C}^2(\R)$ and the map $t\mapsto\tilde{\mu}(t,\cdot,\cdot)$ is continuous of $[0,1]$ to $\mathbf{C}^1$. Therefore, conditions \ref{as:drift_vol_coeffs}-\ref{as:A3} are satisfied.

It remains to verify condition \ref{aas:tc_sing_onedim}. With the definitions $f(x^1)\df e^{-r + x^1}$, $a^{11}(x^2)\df (\tilde{\nu}(1,x^2))^2$, $a^{12}(x^2)\df e^\alpha(\tilde{\nu}\tilde{\sigma}^1)(1,x^2)$, $b^1(x^2)\df r-(1/2)(\tilde{\nu}(1,x^2))^2$, $g(x^1) \df e^{-r}(e^{x^1}-\Gamma)^+$ the pairing $\mathcal{B}_K$ becomes
\begin{align*}
 \mathcal{B}_K[g,\varphi;1] &= \frac{1}{2}\int_K \sum_{k=1}^2 \left(\frac{\Id f}{\Id x^1}\frac{\Id a^{1k}}{\Id x^2}\frac{\partial \varphi}{\partial x^k} +  \frac{\partial }{\partial x^k}\left(\frac{\Id f}{\Id x^1}\frac{\Id a^{1k}}{\Id x^2}\right)\varphi\right)\frac{\Id g}{\Id x^1}
 -2 \frac{\Id f}{\Id x^1}\frac{\Id b^1}{\Id x^2}\frac{\Id g}{\Id x^1}\varphi\ \Id x\\
 &=\frac{1}{2}e^{-r}\int_{K\cap(x^1\geq \log(\Gamma))}e^{x^1}\text{div}\left(\frac{\Id f}{\Id x^1}\frac{\Id a^{11}}{\Id x^2}\varphi,\frac{\Id f}{\Id x^1}\frac{\Id a^{12}}{\Id x^2}\varphi\right) -2e^{x^1} \frac{\Id f}{\Id x^1}\frac{\Id b^1}{\Id x^2}\varphi\ \Id x\\
 &=-\frac{1}{2}(e^{-r}\Gamma)^2 \int_{\hat{K}} \frac{\Id a^{11}}{\Id x^2}\varphi(\log\Gamma,\cdot)\ \Id x^2,
\end{align*}
where $\hat{K}\df K\cap(x^1=\log(\Gamma))$ and the last step follows by a variant of the divergence theorem and the fact that $\Id b^1/\Id x^2 = -(1/2)\Id a^{11}/\Id x^2$.

Since
\begin{equation*}
\frac{\Id a^{11}}{\Id x^2}(\cdot) = 2\left(\tilde{\nu}\frac{\Id \tilde{\nu}}{\Id x^2}\right)(1,\cdot) = 2e^{-\alpha}\left(\nu\frac{\Id \nu}{\Id y}\right)(m+e^{-\alpha}\cdot),
\end{equation*}

it follows from \ref{as:stoch_vol} that for every bounded, open set $K$ in $\R^2$ we can find a function $\varphi=\varphi(x)$ in $\mathbf{W}^1_{p,0}(K)$, for some $p>1$ such that $\mathcal{B}_K[g,\varphi;1]\neq 0$. For example, we can choose an appropriately truncated, shifted and scaled version of the function $\varphi(x) = -|x|$. The result now follows by Theorem \ref{thm:FBMR}.
\end{proof}

\section*{Acknowledgements}
It is a pleasure to thank Dmitry Kramkov for introducing me to the
topic of market completion with derivative securities and for interesting discussions. I would like to thank L\'{e}onard Monsaingeon and Peter Tak\'{a}c for discussions
relating to the presented work and Johannes Ruf for comments on an early version of this paper. 

\appendix
\section{Analyticity as a map for compositions of functions}\label{ap:analytic_comp}
The following lemma is needed for the proof of Theorem \ref{thm:stoch_vol}.
\begin{lemma}\label{lem:analytic_comp}
 Let $S$ be a bounded, open set in $\R$ and the measurable functions $g=g(x) : \R \to \R$, $f = f(t,x) : S\times\R\to\R$ have the property that $g(\cdot)$, $f(\cdot,x)$ are infinitely many times continuously differentiable, $f(t,\cdot)$, $\partial_t f(t,\cdot)$ are twice continuously differentiable and assume there exists a continuous function $C=C(x)$ on $\R$ and constants $\delta,R,D,r,\epsilon >0$ such that
  \begin{equation*}
  \begin{aligned}
  \left| \frac{\partial^k g}{\partial x^k}(x) \right|&\leq \frac{Dk!}{(r+\epsilon|x|)^k},&&x\in\R, \\
  \delta|C(x)|\frac{k!}{R^k}\leq \left| \frac{\partial^k f}{\partial t^k}(t,x) \right|&\leq |C(x)| \frac{k!}{R^k}, &&(t,x)\in S\times\R,
  \end{aligned}
 \end{equation*}
for $k=1,2,\ldots$, then the map $t\mapsto h(t,\cdot) \df (g\circ f)(t,\cdot)$ is analytic as a map of $S$ to $\mathbf{C}(\R)$.
\end{lemma}

\begin{proof}
 To prove the analyticity assertion we have to show the existence of positive constants $L,M>0$, such that
 \begin{equation}\label{eq:ap_est}
  \left\|\frac{\partial^k h}{\partial t^k}(t,\cdot)\right\|_{\mathbf{C}(\R)} \leq \frac{Mk!}{L^k}, \quad t\in S.
 \end{equation}
 By our differentiability hypothesis we may apply the formula of Fa\'{a} di Bruno to obtain
 \begin{equation*}
  \frac{\partial^k h}{\partial t^k}(t,x) = \sum \frac{k!}{\alpha_1!\ldots\alpha_k!} \frac{\partial^{|\alpha|}g}{\partial x^{|\alpha|}}(f(t,x))\left[\left(\frac{\partial f}{\partial t}\right)^{\alpha_1}\ldots \left(\frac{1}{k!}\frac{\partial^k f}{\partial t^k}\right)^{\alpha_k}\right](t,x),
 \end{equation*}
 where the sum is taken over all $\alpha_1,\ldots,\alpha_k$ such that $\alpha_1 + 2\alpha_2 +\ldots + k\alpha_k = k$.
Using our estimates on the derivatives of $g$ and $f$ and Lemma 1.4.1 in \cite[p.~18]{sKrantz2002} we estimate
\begin{align*}
 \left|\frac{\partial^k h}{\partial t^k}(t,\cdot)\right| &\leq k!\frac{D}{R^k}\sum \frac{|\alpha|!}{\alpha_1!\ldots\alpha_k!} \left(\frac{|C(x)|}{r+\epsilon\delta|C(x)|}\right)^{|\alpha|}\\
 &= \frac{k!D|C(x)|}{R^k(r+\epsilon\delta|C(x)|)}\left(1 +\frac{|C(x)|}{r+\epsilon\delta|C(x)|} \right)^{k-1}\\
 &=\frac{D|C(x)|}{r+(1+\epsilon\delta)|C(x)|} \left[\frac{k!}{R^k}\left(\frac{r+(1+\epsilon\delta)|C(x)|}{r+\epsilon\delta |C(x)|}\right)^k\right].
\end{align*}
Taking norms the estimate \eqref{eq:ap_est} follows with $M = D/(1+\epsilon\delta)$ and $L = R\epsilon\delta/(1+2\epsilon\delta)$.
\end{proof}

\bibliography{MarketCompletenessOP_1D.bib}
\bibliographystyle{apalike}
\end{document}